\def\year{2018}\relax
\documentclass[letterpaper]{article} 
\usepackage{aaai18}  
\usepackage{times}  
\usepackage{helvet}  
\usepackage{courier}  
\usepackage{url}  
\usepackage{graphicx}  
\usepackage{amsthm}
\usepackage{amsmath,bm}
\usepackage{bbm}
\usepackage{amssymb}
\usepackage{url}
\usepackage{amsmath}
\usepackage{amsfonts}
\usepackage{dsfont}
\usepackage{graphicx}
\usepackage[ruled,vlined,linesnumbered]{algorithm2e}

\DeclareMathOperator*{\argmax}{argmax}

\newtheorem{theorem}{Theorem}

\newtheorem{lemma}[theorem]{Lemma}

\newcommand{\calG}{\mathcal{G}}
\newcommand{\calI}{\mathcal{I}}

\newcommand{\calS}{\mathcal{S}}
\newcommand{\calT}{\mathcal{T}}
\newcommand{\calU}{\mathcal{U}}
\newcommand{\calX}{\mathcal{X}}
\newcommand{\calZ}{\mathcal{Z}}

\newcommand{\calGG}{\mathds{G}}
\newcommand{\matW}{\mathbf{W}}
\newcommand{\matA}{\mathbf{A}}
\newcommand{\bK}{{\bf k}}
\newcommand{\Gs}{{\bf Gscore}}

\title{SAGA: A Submodular Greedy Algorithm for Group Recommendation}

\author{ Shameem A Puthiya Parambath \\
    QCRI, HBKU, Doha, Qatar\\
    spparambath@hbku.edu.qa\\
    \And Nishant Vijayakumar \\
    Apptopia Inc., Boston, USA\\
    nishant.vijayakumar@gmail.com\\ 
    \And  Sanjay Chawla\\
    QCRI, HBKU, Doha, Qatar\\
    schawla@hbku.edu.qa\\
}

\begin{document}

\maketitle
\begin{abstract}
In this paper, we propose a unified framework and an algorithm for the problem of group recommendation where a fixed number of items or alternatives can be recommended to a group of users.
The problem of group recommendation arises naturally in many real world contexts, and is closely related to the budgeted social choice problem studied in economics.
We frame the group recommendation problem as choosing a subgraph with the largest group consensus score in a completely connected graph defined over the item affinity matrix.
We propose a fast greedy algorithm with strong theoretical guarantees, and show that the proposed algorithm compares favorably to the state-of-the-art group recommendation algorithms according to commonly used relevance and coverage performance measures on benchmark dataset.
\end{abstract}

\section{Introduction}
Recommender systems can be regarded as a particular instance of a personalized information retrieval system where the explicitly or implicitly learned user profiles/features serve the purpose of a `query' in the traditional IR system.
Recommender systems provide effective information exploration by proposing a set of relevant items/alternatives to users.
Many efficient and scalable algorithms exist for personalized recommender systems, carefully tailored for many desirable properties of the recommendation set like relevance, diversity, coverage etc \cite{ekstrand2011collaborative,Steck_11_RecSys,puthiya2016coverage}.
In contrast to personalized recommender systems, a group recommender (GR) system has to provide recommendations for a set of users jointly, returning a set of items in consensus with diverse user preferences.
The user interests might be often conflicting in nature.
A number of real world applications can be formulated within the framework of choosing a subset of items from a given item set to appeal to a group of users.
For example, group discount vendors have to choose a set of items that best appeals to the set of diverse users in a given city; similarly flight operators have to choose a set of movies to play in a plane's entertainment system.
The problem is also closely related to the budgeted social choice problem studied in economics \cite{lu2011budgeted,piotr15_collect}.

In recommender systems, items are assumed to have intrinsic preference scores for each user, which represent the level of affinity of the user to the item.
We assume the preference scores to be on an ordinal scale, typically in the range of 1-5, with higher values indicating stronger preference.
A simple and straightforward approach to group recommendation is to aggregate the preference scores i.e combine the preference scores of individual group members for an item 
 to form a joint ranking of the items and select the top-$\bK$ items \cite{masthoff2011group,baltrunas2010group}.
Most commonly used aggregation functions are \emph{(i)} \textit{average misery}; recommending items with highest average preference scores among the group members, \emph{(ii)} \textit{least misery}; recommending items that maximize the minimum preference scores among users and \emph{(iii)} \textit{plurarility}; recommending items that maximize the number of ``highest preference scores'' \cite{masthoff2011group,amer2009group}.
But such score aggregation strategies completely disregard the user disagreement or the individual user satisfaction for  individual items as discussed in \citeauthor{amer2009group} (\citeyear{amer2009group}).
To get a more balanced recommendation, it is important to consider the disagreements between users of the group for the same item.
Thus the problem can be reformulated as a bi-objective optimization problem where one objective corresponds to the agreement (relevance) of the item to the users and the other objective corresponds to the disagreement between the users for the item.
The general procedure to solve a multi-objective optimization problem is to take the convex combination of the objectives and perform homotopy continuation, also referred as scalarization or weighted-sum approach \cite{EhrgottGandibleuxbook02}.
Similar to aggregation strategies, scalarization based algorithms assume that preference scores for the test items are available, but in practice one has to run a baseline recommender system as a first step to generate the preference scores.

Inspired by the recent work on personalized diverse recommendations \cite{puthiya2016coverage}, in this paper we propose a general framework for algorithms for group recommendations.
We propose a generic objective function by framing the problem as a relevant group interest coverage in an affinity graph defined over the item features.
The crux of our approach is that the item-item and user-user interactions are captured in a single objective.  
Unlike score aggregation strategies, we do not require the preference scores of the test items to be generated beforehand.
We give two instances of the objective where the GR problem reduces to a submodular function maximization problem.
The `diminishing return' property of submodular functions allows us to trade-off the redundancy and relevance, thus making the recommended set containing relevant items covering a wider spectrum of group interest.
An intuitive graphical representation of the underlining idea behind our approach is given in Figure~\ref{fig:art_data}.
The figure contains preferences of a group of three users spanning three item classes, users being represented by the blue '+','*' and '$\times$' symbols.
None of the three users have common item classes, but the group covers all the item classes.
Hence the optimal recommendation strategy is to recommend relevant items from the three item clusters.
The red square is the recommendation made by our proposed algorithm, whereas the baseline algorithms recommendations covered only two item clusters (See Section~\ref{sec:exp} for details).

\begin{figure}[tb!]
\centering
\frame{\includegraphics[width=0.47\textwidth]{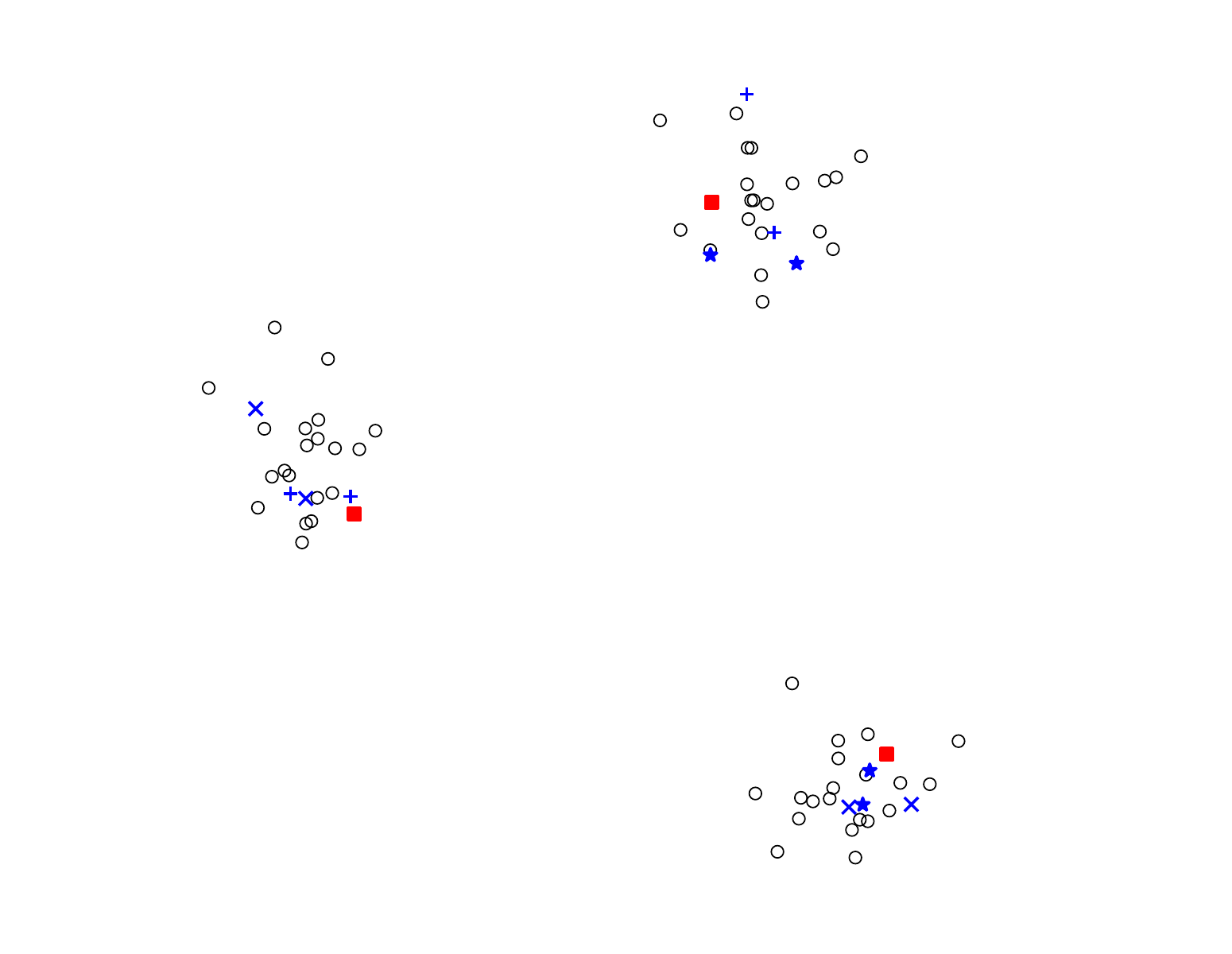}}
\caption{Graphical intuition of the proposed approach}
\label{fig:art_data}
\end{figure}

The remainder of this paper is structured as follows.
We give a brief overview of group recommendation algorithms in Section~\ref{sec:rel_work}.
Section~\ref{sec:group_rmd} describes our framework and relationship to other group recommendation strategies .
We carry out large scale experiments on benchmark data and report the results in Section~\ref{sec:exp}.
We conclude the paper in Section~\ref{sec:conc}.

\section{Related Work}
\label{sec:rel_work}
Most of the earlier work on group recommendation is based on a model where the individual group member attributes are aggregated to generate the group recommendation \cite{masthoff2011group}.
Based on this model, state of the art algorithms can be further classified as profile aggregation based and score aggregation based.
In the first approach, individual user profiles are aggregated to create a group profile.
The group profile can be seen as a single artificial user capturing the group dynamics, and state of the art personalized recommendation algorithm can be used to produce group recommendations for this artificial user.
In the second approach also, the group recommendation is produced in two steps.
In the first step, preference scores for each individual group member for the test items are generated by running a collaborative filtering algorithm.
The second step involves running a score aggregation algorithm on the preference scores obtained in the first step to generate the group recommendations. 
In one of the earliest work on group recommendation, \citeauthor{o2001polylens} (\citeyear{o2001polylens}) proposed an algorithm where the individual user ratings are generated using nearest neighbor models for collaborative filtering, and employing score aggregation strategies like \textit{least misery} and \textit{average misery} to obtain group recommendations.
Similarly, in \cite{kim2010group}, the authors proposed an algorithm by aggregating the individual user profiles to create a group profile and using neighborhood models to generate recommendations for the newly created group profile.
This algorithm comes under the first approach explained above.
\citeauthor{masthoff2011group} (\citeyear{masthoff2011group}) gives an overview and use cases of commonly employed scoring strategies in recommender systems.
However, both the \textit{score aggregation} and \textit{profile aggregation} approaches fail to attribute for the individual user disagreements for different items \cite{amer2009group}.

\citeauthor{amer2009group} (\citeyear{amer2009group}) advocated that a better group recommendation strategy can be devised by considering the disagreements between the individual users in the group for the same item.
They proposed an algorithm where the recommendation set is obtained by solving a discrete optimization problem.
The objective function is a linear combination of an agreement term and a disagreement term.
It is an instantiation of the scalarization algorithm for the multi-objective optimization problem.
The proposed algorithm requires one to run a baseline personalized recommender system to generate preference scores for every user.
Moreover, many variants of the above discussed scalarization based group recommendation schemes are proposed in domain specific settings and we direct the reader to the survey papers by \citeauthor{jameson2007recommendation} (\citeyear{jameson2007recommendation}) and \citeauthor{boratto2010state} (\citeyear{boratto2010state}), and the references therein.
One of the difficulties in developing group recommendation algorithms is the lack of proper evaluation strategies.
In \citeauthor{baltrunas2010group} (\citeyear{baltrunas2010group}) , the authors propose a methodology for offline evaluation of group recommendation algorithms.
Recent works make use of expensive online evaluation strategies \cite{Sanchez2014,boratto2016discovery,Liu2016,khoshkangini2016self}

In economics, group recommendation is studied as a special case of the budgeted social choice problem \cite{lu2011budgeted,piotr15_collect}.
The budgeted consensus recommendation is often considered as a trade-off between pure personalized recommendation and pure consensus recommendation.
In group recommendation, the budget is specified as a function of the recommendation size, keeping the other cost components constant.
Most of the work on budgeted social choice theory assumes that user preferences for the items are strict, i.e. every user derives benefit from at most one item that has the highest preference score (utility score), and preference scores are often derived using positional scoring rules like Borda count.
\citeauthor{lu2011budgeted} (\citeyear{lu2011budgeted}) proposed a greedy algorithm using knapsack heuristics, but it worked very poorly on group recommendation tasks with positional scoring rules.
\citeauthor{piotr15_collect} (\citeyear{piotr15_collect}) extended the theoretical framework for limited choice models with positional scoring rules using the ordered weighted average (OWA) operators.

Unlike the above mentioned approaches, our approach proceeds in a single step by carrying out preference estimation implicitly and thus avoids running a baseline personalized recommender system to generate the \textit{unseen} test preference scores.
The trade-off between the agreement and the disagreement components is dealt with by the exact definition of the \emph{user saturation function}, \emph{item saturation function} and the affinity functions, instead of explicitly trading-off of both the terms.
Our work is inspired by \citeauthor{puthiya2016coverage} (\citeyear{puthiya2016coverage}) where the authors propose a personalized diverse recommendation algorithm.
Nonetheless, our work differs from theirs significantly.
We introduce group consensus score as a function of the \textit{user saturation function} and the \textit{item saturation function}.
The exact notion of \textit{user saturation function} is motivated by the \textit{item saturation function} defined in \citeauthor{puthiya2016coverage} (\citeyear{puthiya2016coverage}), to capture the coverage of an item observed by a user.
Our proposed algorithm is more general and the algorithm proposed in \citeauthor{puthiya2016coverage} (\citeyear{puthiya2016coverage}) is a special case of our approach in that when the group consists of a single user, the \textit{item saturation function} is a concave function and the \textit{user saturation function} is the identity function.
By defining the \textit{user coverage function} to be submodular, we aim to \emph{(i)} recommend items which are relevant to the majority of the group and \emph{(ii)} discourage the recommendations to be centered around a few user.

\section{Group Recommendation}
\label{sec:group_rmd}
We are given a set $\calX$ of $n$ items and a set $\calU$ of $m$ users.
The set of users forms a set of groups $\calGG$.
The group dynamics can evolve over time, and a user may or may not be a member of a group.
Also, a user can be part of just a single group or multiple groups.
We also assume the existence of an affinity function over the item space $h: \calX \times \calX \rightarrow \Re_+$.
Similarly, we assume the existence of an affinity function over the user space $c: \calU \times \calU \rightarrow \Re_+$.
We assume that the affinity functions $h$ and $c$ are non-negative.
There is no assumption regarding the symmetricity or transitivity of the affinity functions.
Each user has a preference score (utility score or relevance score) for each of the item denoted using an ordinal number, typically in the range 1-5 with higher value indicating stronger preference.
A user can have same preference score for different items.

Given the item set $\calX$ and the corresponding item-item affinity matrix $\matW$, we can view the pair $(\calX,\matW)$ as a complete graph where the edge between nodes $i,j \in \calX$ is weighted according to the value $\matW_{ij}$.
Given a group $\calG \in \calGG$, the user-user affinity matrix $\matA$ and the observed preference scores for the set of users in the group $\calG$ for the item set $\calI \subset \calX$, the item set $\calI$ defines a subgraph of the complete graph $(\calX,\matW)$.
We can define the group recommendation task as retrieving a subset of $\bK$ items from the set $\calZ = \calX \setminus \calI$ such that it \emph{(i)} \textit{covers} most of the observed items and has high preference scores in the set $\calI$ and \emph{(ii)} \textit{covers} a larger spectrum of users in the group $\calG$.
We introduce group consensus score to capture the above two aspects of the group recommendation.
Formally, we define the group consensus score for a set of items $\calS \subseteq \calZ$ with cardinality $\bK$ with respect to the set $\calI$ for the given group $\calG$ as,

\begin{equation}
\begin{split}
\Gs(\calG,\calI,\calS) &= \\
 \sum\limits_{u \in \calG} g_u&\Big(\sum\limits_{l \in \calG \setminus \{u\}}\matA_{ul} \sum\limits_{i \in \calI_u} r_u^if\big(\sum\limits_{j\in\calS} \matW_{ij}\big)\Big)
\end{split}
\end{equation}
where $g_u:\Re_+ \rightarrow \Re$ is the \textit{user saturation function},  $f:\Re_+ \rightarrow \Re_+$ is the non-negative \textit{item saturation function}, $\calI_u$ is the set of observed items in $\calI$ for the user $u$, and $r_u^i$ is the observed preference score for the item $i$ by user $u$.
The rationale behind the objective function is to find the items similar to the observed relevant items which appeals for the entire group, as demonstrated by higher value for $\Gs$.
By choosing proper \textit{item saturation} and \textit{user saturation} functions, we aim to produce a consensus set of recommendations.
Given a fixed \textit{user saturation function} $g_u$ and \textit{item saturation function} $f$, a set of $\bK$ items for the group $\calG$ with the highest group consensus score can be obtained by solving the optimization problem
\begin{equation}
\argmax_{\substack{\calS \subseteq \calZ \\ \lvert\calS\lvert \leq \bK}} \Gs(\calG,\calI,\calS)
\label{eq:opt_set}
\enspace.
\end{equation}
Before further discussion into the design details of the \textit{saturation functions} $g_u$ and $f$, we describe the generality of the objective function \eqref{eq:opt_set}.
It can be shown that many well-known group recommendation strategies are special cases of the objective function \eqref{eq:opt_set}. 

\subsection{Special Cases}
For different formulations of the \textit{user saturation function} $g_u$, the optimization problem in \eqref{eq:opt_set} corresponds to \textit{score aggregation} strategies currently employed in group recommendation tasks.

\subsubsection{Identity item saturation function}
Here, we see that by fixing the \textit{item saturation function} ($f$) to be the identity function i.e. $f(x) = x$ and varying the \textit{user saturation function} $g_u$ gives rise to many score aggregation strategies.
We will also assume that $\matA$ is the indicator matrix for the given group $\calG$ i.e. $\matA_{ij} = 1 $ if $i$ and $j$ are in group $\calG$ , zero otherwise.

\smallskip\noindent\textit{Least Misery:} If $g_u$ is set as the constant function $\min$ i.e. $g_u(x) = \min(x), \forall u \in \calG$, then the optimization problem corresponds to selecting $\bK$ items which maximize the lowest preference scores on the test set, among all members in the group.
This is equivalent to the \textit{least misery} aggregation strategy \cite{masthoff2011group}.

\smallskip\noindent\textit{Most Pleasure:}
Similarly, by setting $g_u$ to the $\max$ function i.e. $g_u(x) = \max(x), \forall u \in \calG$, the optimization problem in \eqref{eq:opt_set} becomes the \textit{most pleasure} scoring strategy \cite{masthoff2011group}.
$\bK$ items that maximize the highest preference scores for the items in the test set, among all members in the group are recommended.

\smallskip\noindent\textit{Average Misery:}
If $g_u$ is set as the constant function $g_u(x) = \frac{1}{\lvert \calG \lvert}, \forall \calG \in \calGG$, then the optimization problem reduces to the average misery approach, where the items that maximize the average preference scores among all the members in the group are recommended \cite{masthoff2011group}.

\smallskip\noindent\textit{Plurality:}
If $g_u$ is the indicator function such that it takes the value one when the preference score is the highest (ties are broken arbitrary) and zero otherwise, then the optimization problem corresponds to the plurality scheme \cite{masthoff2011group}.

\smallskip\noindent\textit{Ordered Weighed Average Operators:}
If the inner term $\sum\limits_{i \in \calI_u} r_u^if\big(\sum\limits_{j\in\calS} \matW_{ij}\big)$ is assumed to be in the sorted order, and for a general $g_u: \Re_+ \rightarrow \Re_+$ , the group consensus function corresponds exactly to the ordered weighted average (OWA) operators studied by \citeauthor{piotr15_collect} (\citeyear{piotr15_collect}).
In this case, many different functional formulations of $g_u$ will result in different group recommendation strategies like $\bK$-best OWA, Hurwicz OWA etc (See \citeauthor{piotr15_collect} (\citeyear{piotr15_collect}) for details).

\subsubsection{Identity user saturation function}
Here, we fix the \textit{user saturation function} ($g_u$) to be the identity function i.e. $g_u(x) = x$ and vary the \textit{item saturation function} $f$.

\textit{Diverse Personalized Recommendation:}
In case of single user groups ($\matA$ will be the identity matrix), when the \textit{item saturation function} is concave, the objective function reduces to the one studied by \citeauthor{puthiya2016coverage} (\citeyear{puthiya2016coverage}) for diverse personalized recommendations.
For other functional forms of $f$,  we get different personalized recommendation schemes \cite{puthiya2016coverage}.

\textit{Time Evolving Group Dynamics:}
The individual user preferences are dynamic and it can evolve over time.
User's interest for activities like shopping, watching movies etc, can fade or strengthen over time.
It is important to adjust the parameters of the algorithm to reflect the evolving group dynamics by lowering the significance of the ``fading'' user and highering the significance of the ``strengthening'' user.
One way to achieve this is by re-weighting the individual user contribution in a group when recommending new items.
This can be done by defining the \textit{user saturation function} $g_u$ as $g_u(x) = \frac{x}{t_u}$, where $t_u$ is the ``\textit{transportation time}'' of user $u$ in the given group.
One can define the ``\textit{transportation time}'' as a function of user's activity time in the group where a frequent loyal user will have low \textit{transportation time} whereas occasional user will have high \textit{transportation time}.

\subsection{Design of Saturation Function}
Ideally, we want the recommended set to include relevant yet items that appeals to a larger spectrum of users in the group.  
We design \textit{saturation functions} such that the $\Gs$ has the ``diminishing return'' property.
In our settings, the ``diminishing return'' property ensures that once a relevant item with respect to one (or few) user(s) is added, the marginal gain of adding another relevant item for the same user(s) is less than adding a relevant item with respect to the remaining users.
In particular, the ``diminishing return'' property of the \textit{item saturation function} helps to select relevant yet diverse items, whereas the ``diminishing return'' property of the \textit{user saturation function} encourages user coverage. 
The ``diminishing return'' is the characterizing property of submodular functions, and we choose the \textit{saturation functions} to be such that group consensus score is a submodular function.

We discuss two settings.
In the first settings, we propose a simple extension of the algorithm proposed in \citeauthor{puthiya2016coverage} (\citeyear{puthiya2016coverage}) by choosing $g_u$ to be the identity function.
\begin{lemma}
For any concave function $f$ and the identity function $g_u$, $\Gs$ is a submodular function (with respect to inclusion in $\calS$). If $f$ is monotonic then $\Gs$ is also monotonic.
\label{th:lemma1}
\end{lemma}
The proof of Lemma \ref{th:lemma1} can be found in \citeauthor{bach2013learning} (\citeyear{bach2013learning}).
In the second setting, we choose $g_u$ to be a concave function.

\begin{lemma}
For any monotone concave function $f$ and any concave funtion $g_u$, $\Gs$ is a submodular function (with respect to inclusion in $\calS$). If $g_u$ is monotonic then $Gscore$ is a monotonic submodular function.
\label{th:lemma2}
\end{lemma}
\begin{proof}
By Lemma\ref{th:lemma1}, $\sum\limits_{l \in \calG \setminus \{u\}}\matA_{ul} \sum\limits_{i \in \calI_u} r_u^if\big(\sum\limits_{j\in\calS} \matW_{ij}\big)$ is a submodular function of $\calS$.
Let us denote this term as $z(\calS)$.
It is sufficient to prove that $g(z(\calS))$ (we omit the subscript for convenience) is a submodular function.
Due to the concavity of $g$ and monotonicity of $z$ (due to the monotonicity of $f$), $\forall \calS_1, \calS_2 \text{ such that } \calS_1 \subseteq \calS_2 \subseteq \calZ$ and $\forall e \in \calZ$ such that $e \notin \calS_2$,
\begin{equation*}
g(z(\calS_1 \cup e)) - g(z(\calS_1)) \geq g(z(\calS_2 \cup e)) - g(z(\calS_2))
\end{equation*}
which completes the proof.  \hspace*{\fill}\end{proof}

\subsection{Accelerated Greedy Algorithm}
For any monotone concave function $f$ and any concave funtion $g_u$, the optimization problem \eqref{eq:opt_set} becomes a submodular maximization problem.
The problem is NP-hard and no algoithms exist to solve it exactly \cite{nemhauser1978analysis}.
The general strategy to solve the submodular maximization problem is based on the greedy heuristic as given by \citeauthor{nemhauser1978analysis} (\citeyear{nemhauser1978analysis}).
The algorithm iteratively selects an element from the ground set such that it gives the maximum value for the incremental update value $\Gs(\calG,\calI,\calS \cup \{e\}) - \Gs(\calG,\calI,\calS)$ at each iteration (the ties are broken arbitrarily).
In case of monotonic submodular functions with cardinality constraints, \citeauthor{nemhauser1978analysis} (\citeyear{nemhauser1978analysis}) gives a worst case lower bound on the optimality gap between the optimal solution and the greedy solution as given in Theorem~\ref{th:greedy_bound}.
In fact, similar bound holds for any monotonic submodular functions with matroid constraints and cardinality is a special case of matroid constraint \cite{fisher1978analysis}.

\citeauthor{minoux1978accelerated} (\citeyear{minoux1978accelerated}) proposed a faster ``accelerated'' version of the greedy algorithm using the fact that $g(\calS_{j}\cup \{l\}) -g(\calS_{j}) \leq g(\calS_i \cup \{k\}) - g(\calS_i) \implies g(\calS_i \cup \{l\}) \leq g(\calS_i \cup \{k\}), \forall j < i$.
The algorithm carries out `lazy' updates thus avoiding a complete scan to find the next item to recommend \cite{leskovec2007cost}.
By using priority queues, the retrieval at line 5 and updation at line 10 in Algorithm\ref{alg:acc_greedy} can be achieved in constant and logarithmic time respectively.
Previous experimental results on large scale datasets showed that the accelerated greedy algorithm gives a substantial speedup boost \cite{leskovec2007cost}.
Our proposed accelerated greedy algorithm for group recommendation which we call \textbf{S}ubmodul\textbf{A}r \textbf{G}roup recommendation \textbf{A}lgorithm (\textbf{SAGA}) is  given in Algorithm~\ref{alg:acc_greedy}.

\begin{theorem}[\citeauthor{nemhauser1978analysis} (\citeyear{nemhauser1978analysis})]
For a non-decreasing submodular function $f$, let $\calS^*$ be the optimizer of \eqref{eq:opt_set} and $\hat\calS$ be the set returned by the greedy algorithm iteratively adding to the current solution the item that provides the highest gain in the objective function,
then
$\qquad f(\hat{\calS}) \geq \Big(1-(1-\frac{1}{\bK})^\bK\Big) \,f(\calS^*) \geq (1-\frac{1}{e})\,f(\calS^*)$
\label{th:greedy_bound}
\end{theorem}

\begin{algorithm}[tb!]
   \caption{SAGA: Submodular Greedy Group Recommendation Algorithm}
   \label{alg:acc_greedy}
    \SetKwData{Left}{left}\SetKwData{This}{this}\SetKwData{Up}{up}
    \SetKwFunction{Union}{Union}\SetKwFunction{FindCompress}{FindCompress}
    \SetKwInOut{Input}{Input}\SetKwInOut{Output}{Output}
    \Input {set of items $\calX$, set of observed items $\calI$, group $\calG$, affinity matrix $\matW$, \# of recommendations \bK}
    $\calX = \calX \setminus \calI, \calS = \emptyset$ \;
    \For {$i \in \calX$}  {
        $\Delta(i) = Gscore(\calG,\calI,\{i\})$ \tcp*[r]{compute and store the marginal gain for each item in a priority queue}
    }
    \Repeat {$\lvert \calS \lvert = \bK$} {
    $i^* = \argmax_{i \in \calX} \Delta(i)$ \;
    $\delta = Gscore(\calG,\calI,\calS \cup \{i^*\}) - Gscore(\calG,\calI,\calS)$ \;
    $\Delta(i^*) = \delta$ \;

    \If { $\delta < \max_{i \in \calX \setminus \{i^*\}} \Delta(i)$ }  {
        goto 5
    }
    $\calS = \calS \cup \{i^*\}$ \;
    $\calX = \calX \setminus \{i^*\}$ \;
    }
    \Output{ set of recommendations $\calS$}
\end{algorithm}

\section{Experiments}
\label{sec:exp}
In this section, we describe the experimental setup and report the results using user groups generated from the \textit{1M} MovieLens\footnote{\url{http://grouplens.org/datasets/movielens/}} dataset.
A major issue with GR research is the difficulty of evaluating the effectiveness of the recommendations \cite{baltrunas2010group}.
Online evaluation is expensive and it  can be performed only on a very limited set of users.
Here we follow the offline experimental protocol proposed in \citeauthor{baltrunas2010group} (\citeyear{baltrunas2010group}).
We compare our proposed algorithm against two state-of-the-art group recommendation algorithms \cite{o2001polylens,amer2009group}.

\subsection{Group Formation}
The MovieLens dataset contains ratings on a scale of 1 to 5 of 6040 users for 3670 movies.
In our experiments we use users only with a minimum of 100 ratings and ended up with a set of 2945 users and 3670 movies.
Following the previous work by \citeauthor{amer2009group} and \citeauthor{baltrunas2010group}, we create groups corresponding to two common real life scenarios \emph{(i)} random groups and \emph{(ii)} similar groups.
Random groups are generated by randomly sampling a fixed number of users whereas similar groups are generated by randomly sampling from the user set provided that the individual user-user similarity is greater than a threshold value.
We set 0.60 as the similarity threshold value.
The detailed description of the group statistics is given in Table\ref{tab:group_desc}.

\begin{table}
  \centering
  \begin{tabular}{l*{4}{c}}
    \# members & 2 & 4 & 6 & 8 \\
    \hline
    random    & 294 & 146 & 98 & 72  \\
    similar   & 190 &  40 & 18 & 10  \\
    \hline
  \end{tabular}
    \caption{Number of Groups}
    \label{tab:group_desc}
\end{table}

\subsection{Protocol}
We carried out holdout validation by randomly selecting 30\% of the item set and marking it as unrated wherever the rating values are observed.
To reduce the variability in the result, the procedure is carried out five times and the reported results are the average values over the five selections.
The affinity and similarity values are calculated based on the item and user features.
The item and user features are extracted from the rating matrix using the non-negative matrix factorization based on weighted-regularized non-negative alternating least squares algorithm \cite{steck2013evaluation}.
In our experiments, the dimension of the item/user feature space is set to 150.
We used the rbf kernel as the item afinity function $h$, i.e. $\matW_{ij} =\exp(-\gamma||x_i - x_j||^2)$ where $x_i$ and $x_j$ are the $i^{th}$ and $j^{th}$ item features.
The $\gamma$ value is chosen by running the algorithm for a set of values in the range $\{2^{-3},\cdots,2^3\}$ in multiples of two, and the reported results are for the best $\gamma$ value for the respective algorithms.
For the \textit{item saturation function} $f$, we use natural logarithm $f(x) = \ln(1+x)$, and for the \textit{user saturation function} we experimented with two settings \emph{(i)} identity function $g_u(x) = x$ and \emph{(ii)} concave function $g_u(x) = \sqrt{x}$.
The user affinity matrix $\matA$ is defined as the cosine of the user feature vector i.e. $\matA_{ij} = \cos(y_i,y_j)$ where $y_i$ and $y_j$ are the $i^{th}$ and $j^{th}$ user features.
\subsection{Baselines}
We used two baselines to compare against the performance of our proposed algorithm.
First baseline is the popular aggregation strategy \textit{Average Misery} (AM) and the second baseline is the scalarization based algorithm proposed by \citeauthor{amer2009group} (\citeyear{amer2009group}) called group recommendation with fully materialized disagreement list (FM).
Both AM and FM require unobserved rating values to be generated beforehand, and we use the non-negative matrix factorization (used for the user/item feature extraction) to generate the unobserved rating values.
\subsubsection{AM}
\textit{Average Misery} is a popular \textit{score aggregation} strategy where the first $\bK$ items with the highest average score for the group is selected.
Formally, \textit{Average Misery} algorithm selects an item $i^*$ such that

\begin{equation}
i^* \in \argmax_{i \in \calZ} rel(\calG,i)
\enspace.
\end{equation}
where $rel(\calG,i)$ is defined as the sum of the preference scores for the item by the users in the group $\calG$, i.e. $rel(\calG,i) = \sum\limits_{u \in \calG} r_u^i \enspace$.

\subsubsection{FM}
In FM, an agreement term and a disagreement term are linearly combined using a trade-off parameter, and for each value of the trade-off parameter the discrete optimization problem is solved using thresholding \cite{amer2009group}.
Formally, the FM algorithm selects an item $i^*$ such that

\begin{equation*}
i^* \in \argmax_{i \in \calZ} ~\lambda \times rel(\calG,i) + (1-\lambda)(1-dis(\calG,i))
\end{equation*}
where $rel(\calG,i)$ is as defined above, and $dis(\calG,i)$ is the average pair-wise disagreement for the item $i$ for the group $\calG$.
Formally,

\begin{equation*}
dis(\calG,i) = \frac{2}{\lvert \calG \rvert (\lvert \calG \rvert - 1)} \sum_{\substack{u,v \in \calG \\ u \neq v}}\lvert r_u^i - r_v^i \rvert
\end{equation*}

\subsection{Performance Measure}
In our experiments, we measure the relevance and coverage of the recommended items to the group members. 
We use Discounted Cumulative Gain (DCG), a very popular performance measure widely used in information retrieval tasks to measure the utility of the recommended items to a user \cite{jarvelin2002cumulated}, and Popularity Stratified Recall (PSR), which measures the coverage/popularity-bias of the recommended items \cite{Steck_11_RecSys}.

\smallskip\noindent\textit{Discounted Cumulative Gain (DCG):}
DCG is used in the context of ranking.
It measures the quality of a ranked list by the sum of the graded relevance discounted by the rank of the item.
In our experiments, we used the below formulation:
\begin{equation*}
  \sum \limits_{i \in \calS} \frac{2^{r_i}-1}{log(i+1)} 
  \enspace,
\end{equation*}
where $r_i$ is the \emph{graded} relevance score of the $i^{th}$ item.
In our settings, the $i^{th}$ item is the $i^{th}$ item entering $\calS$ for the greedy algorithm.
We compute the average values of the DCG, over all the users in any group, and the reported results are average over the total groups.
Higher values for DCG is desired, as it indicates more relevant recommendations appears at the top of the recommendation list.

\smallskip\noindent\textit{Popularity Stratified Recall (PSR):}
A good recommendation strategy should cover items which are relevant to the users but rare.
PSR is proposed to measure the ability of a recommender system to recommend items from the tail of the item-popularity distribution.
The submodular formulation of our objective increases the chance of recommending unpopular relevant items.
Formally, PSR is defined as:
{
\begin{equation*}
\frac{1}{\lvert \calG \rvert}\frac{\sum\limits_{u \in \calG}\sum\limits_{i \in \calS_u^{+}} \Big(\frac{1}{N_i^+}\Big)^\beta}{\sum\limits_{u \in \calG}\sum\limits_{i \in \calT_{u}} \Big(\frac{1}{N_i^+}\Big)^\beta}
\enspace,
\end{equation*}
where $\calS_u^{+}$ is the set of recommended items that are known to be relevant for user $u$ (among the $\bK$ recommended items),
${\calT}_u$ is the set of items in the test set that are known to be truly relevant for user $u$,
}
$N_i^+$ is the number of relevant ratings for item $i$ in the test set and $\beta \in (0,1)$ is a hyperparameter which adjusts for the popularity bias.
We used $\beta=0.5$ in our experiments.
Higher values for PSR indicates more relevant unpopular recommendations appears at the top of the recommendation list.

\subsection{Results \& Discussion}
\subsubsection{Random Groups}
\begin{figure}[tb!]
\centering
\begin{tabular}{@{}c@{}}
\includegraphics[keepaspectratio=true,scale=0.59]{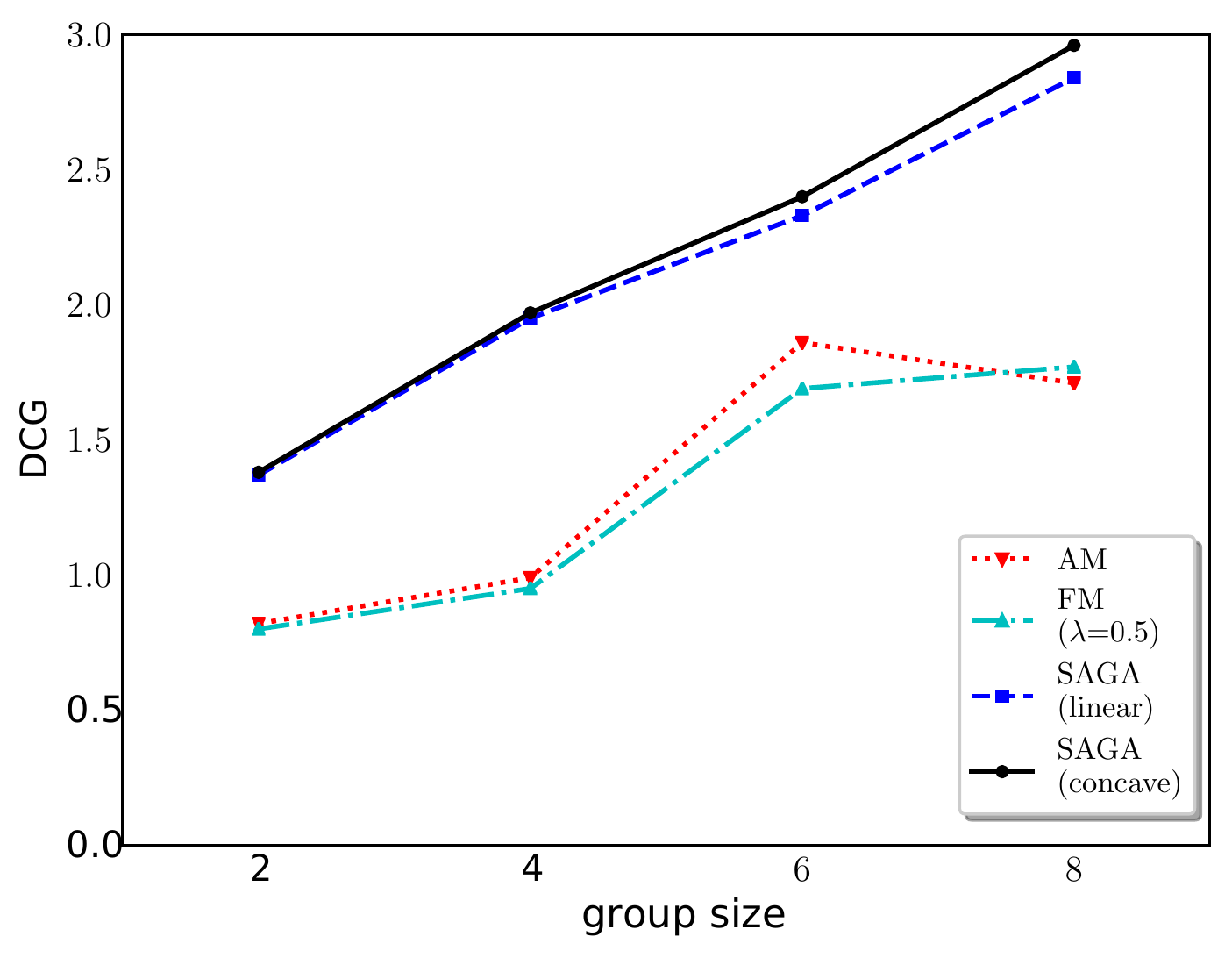}
\end{tabular}
\caption{DCG as function of group size (random)}
\label{fig:dcg_rg_gz}
\end{figure}

\begin{figure}[tb!]
\centering
\begin{tabular}{@{}c@{}}
\includegraphics[scale=0.59]{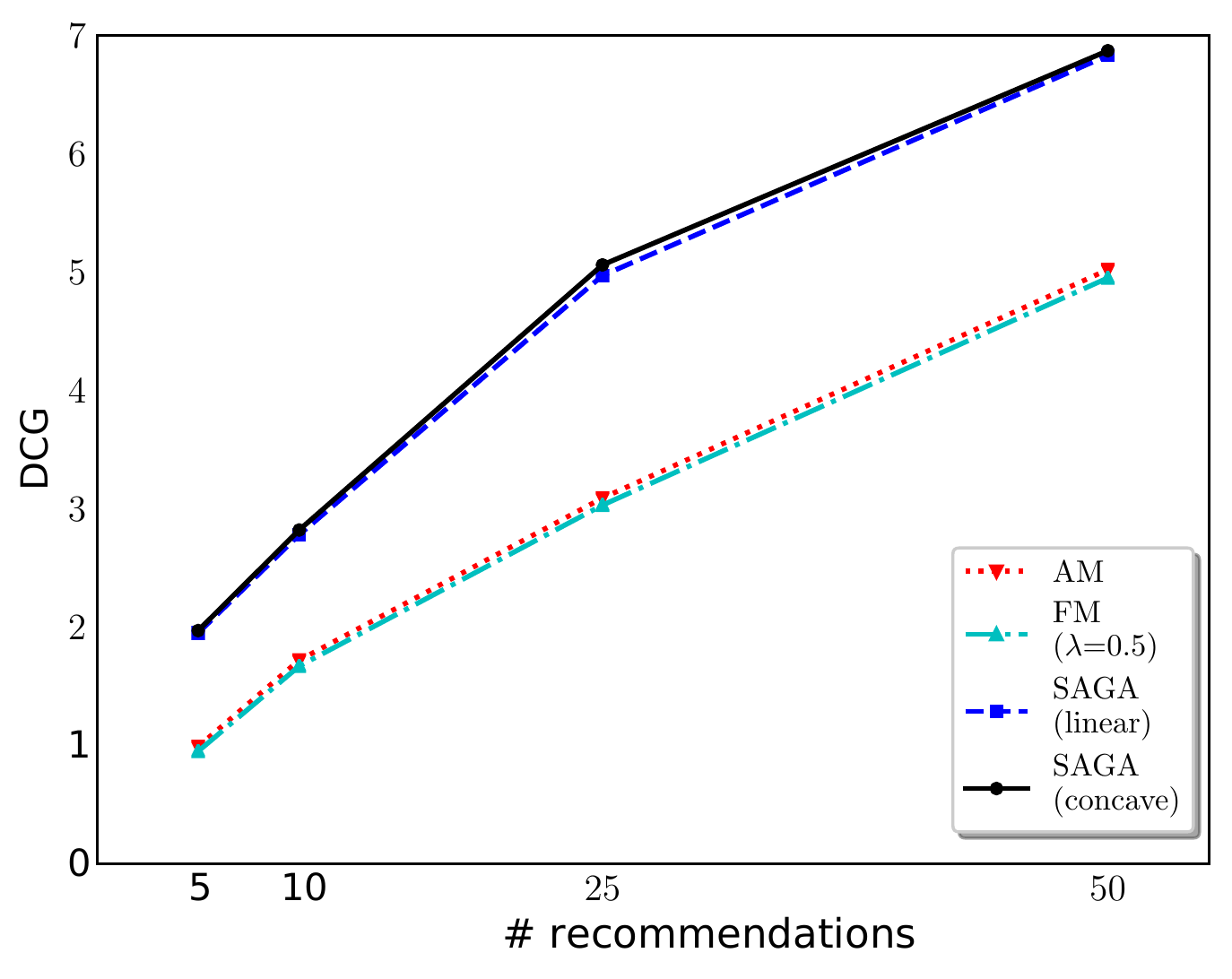}
\end{tabular}
\caption{DCG as function of \# recommendation (random)}
\label{fig:dcg_rg_rz}
\end{figure}

The results for the random groups are given in Figures~\ref{fig:dcg_rg_gz} and \ref{fig:dcg_rg_rz}.
Figure~\ref{fig:dcg_rg_gz}  depicts the performance of different algorithms as a function of group size ($\lvert\calG\lvert$), for a fixed number of recommendations ($\bK$ = 5) for random groups.
The performance of AM and FM algorithms are very similar with AM performing marginally better than FM for groups of smaller sizes.
In effect, the leverage of employing disagreements between group members for recommendation is insignificant in case of random groups.
Similarly, the performance of the linear and concave versions of the proposed SAGA algorithm are very similar.
Howbeit, SAGA ,in general, performs significantly better than both AM and FM.

Figure~\ref{fig:dcg_rg_rz} illustrates the performance of different algorithms as a function of recommendation size ($\bK$), for a fixed group size $(\lvert \calG\rvert= 4)$.
The plot follows the same pattern as in the case of varying group sizes for random groups.
The performance of AM and FM are very close to each other whereas the performance of linear and concave version of the SAGA algorithm are close to each other.
Yet, the two formulations of SAGA give superior DCG values compared to both AM and FM.
For different values of the group size ($\calG \in \{10, 25, 50\}$), we observed the same pattern.
Due to space limitation, we omitted the corresponding plots.

\begin{figure}[tb!]
\centering
\begin{tabular}{@{}c@{}}
\includegraphics[keepaspectratio=true,scale=0.58]{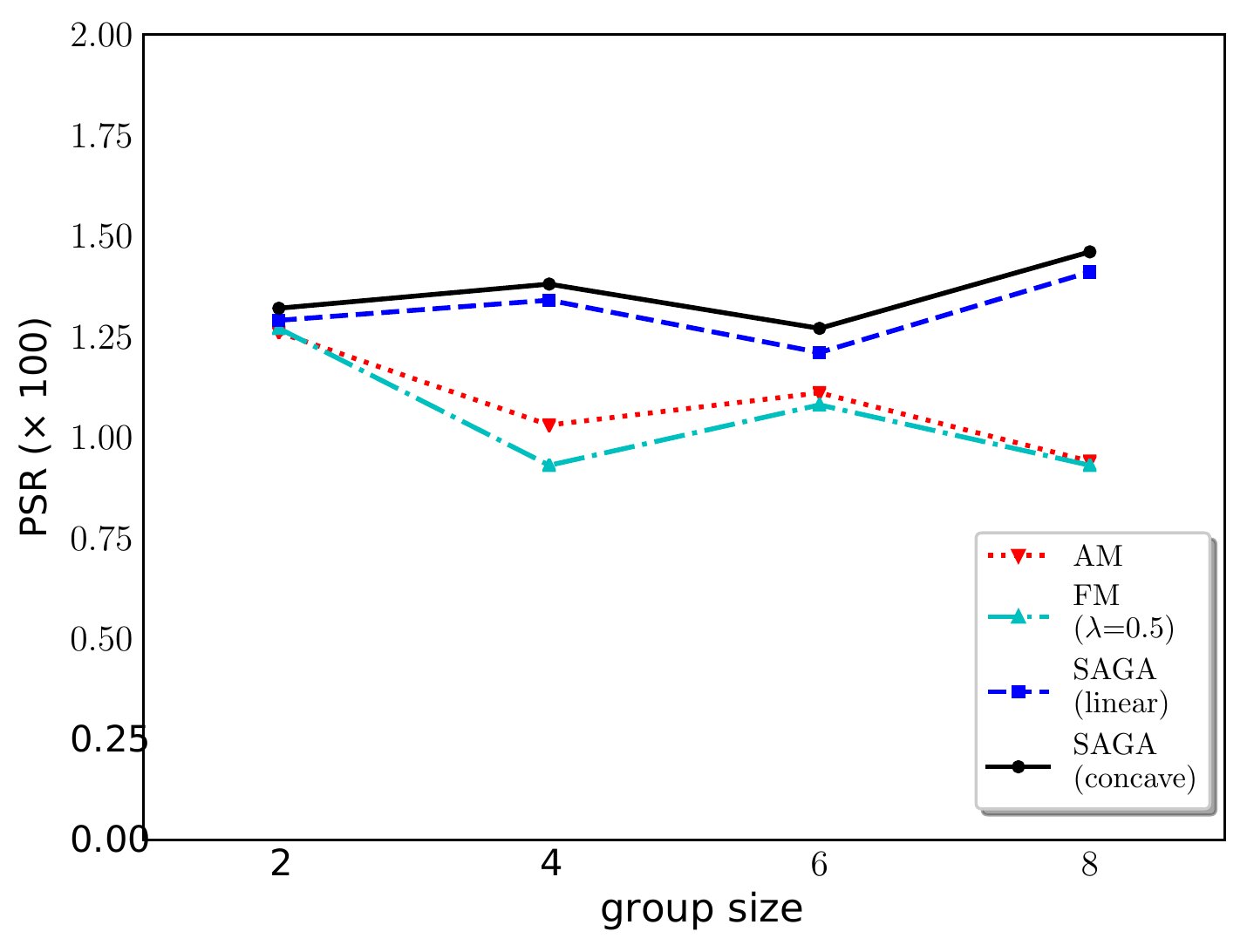}
\end{tabular}
\caption{PSR as function of group size (random)}
\label{fig:sr_rg_gz}
\end{figure}

\begin{figure}[tb!]
\centering
\begin{tabular}{@{}c@{}}
\includegraphics[keepaspectratio=true,scale=0.59]{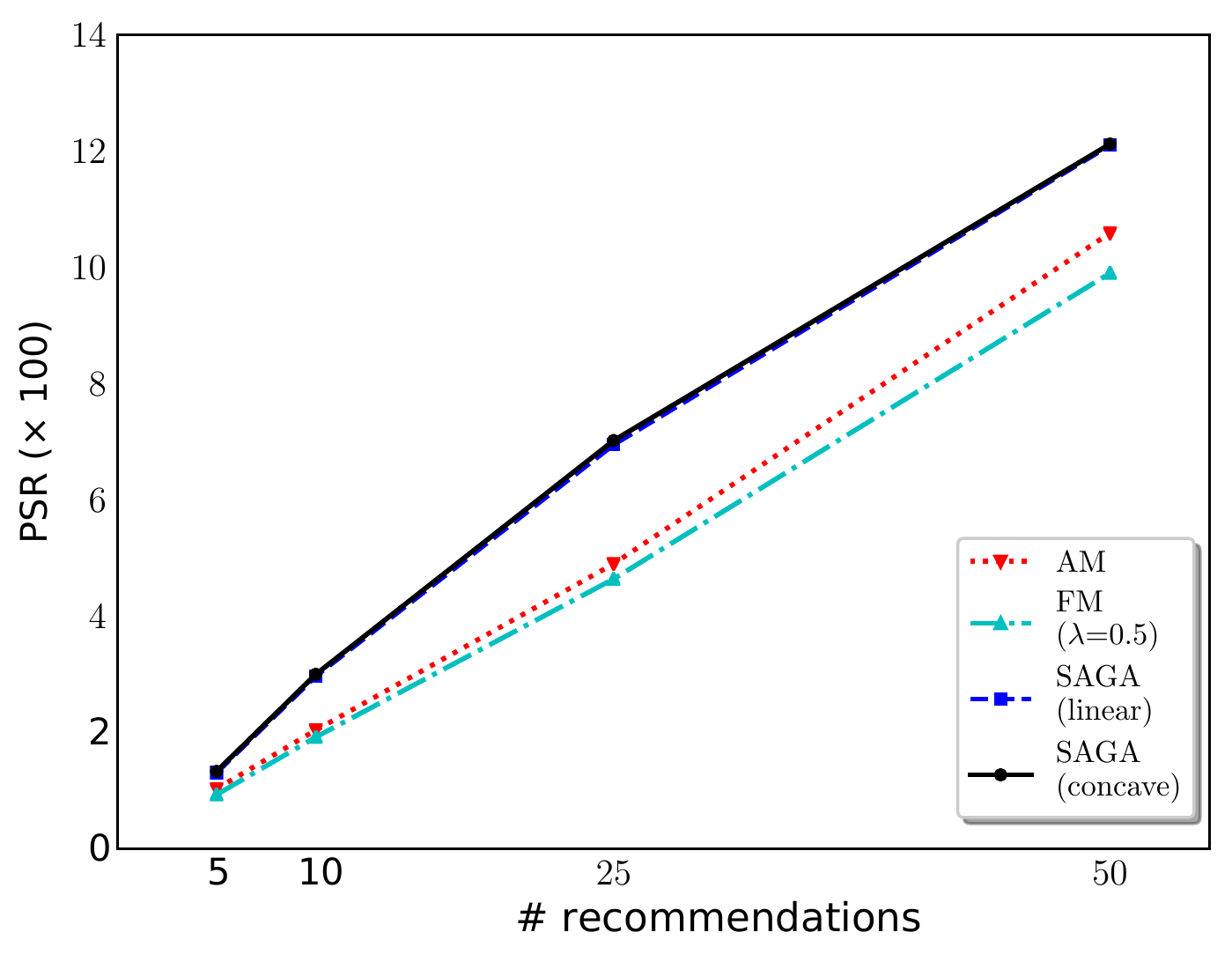}
\end{tabular}
\caption{PSR as function of \# recommendation (random)}
\label{fig:sr_rg_rz}
\end{figure}
Figures~\ref{fig:sr_rg_gz} and \ref{fig:sr_rg_rz} shows the coverage performances of different algorithms as functions of group size ($\lvert\calG\rvert$) and number of recommendations ($\bK$).
The AM and FM algorithms perform equally but sub-optimal.
The SAGA linear and concave versions performs better than AM and FM algorithms with the SAGA concave version performs only marginally better than the linear version.
Thus SAGA algorithms recommend relevant but novel items.

\subsubsection{Similar Groups}
\begin{figure}[tb!]
\centering
\begin{tabular}{@{}c@{}}
\includegraphics[keepaspectratio=true,scale=0.59]{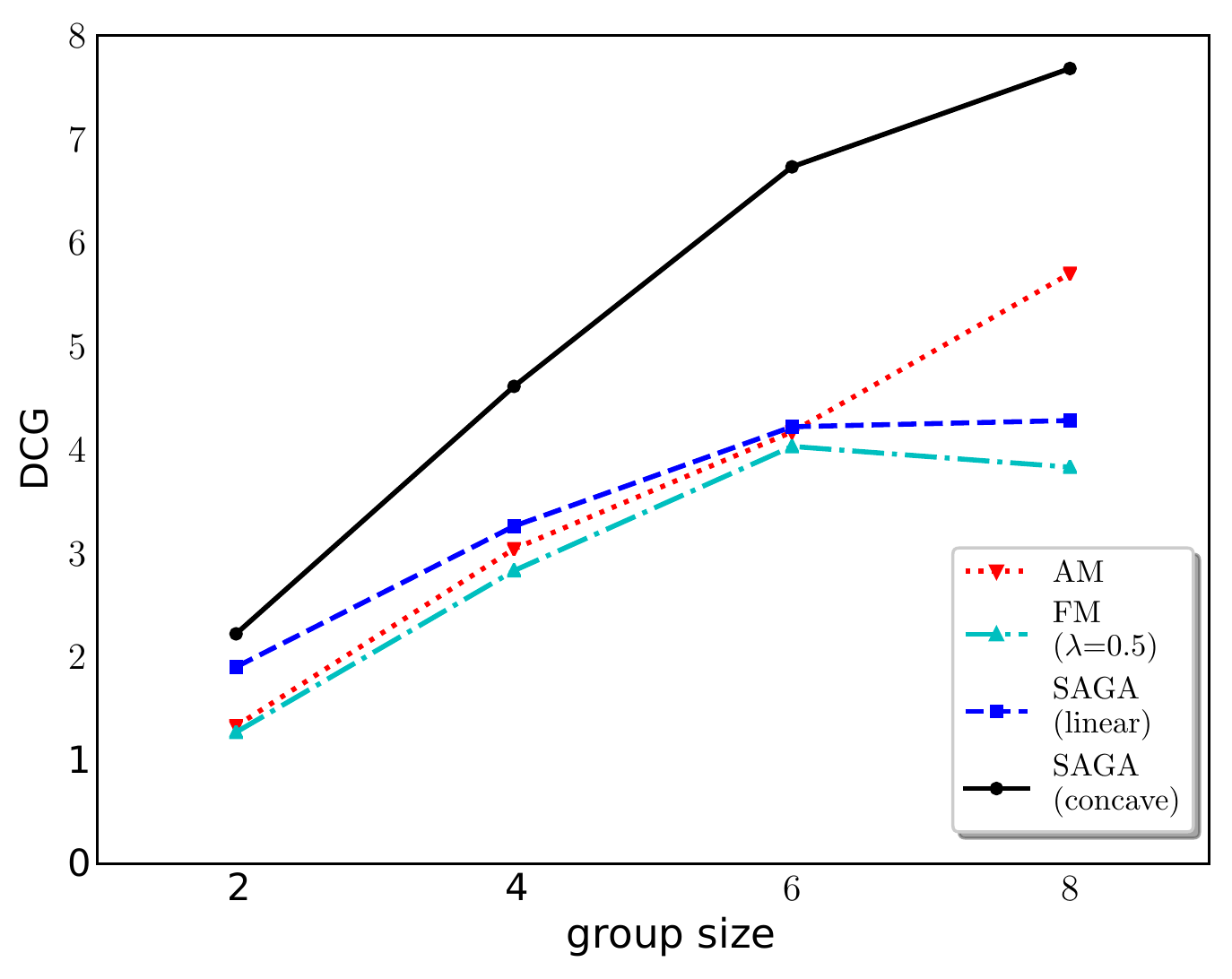}
\end{tabular}
\caption{DCG as function of group size (similar)}
\label{fig:dcg_sg_gz}
\end{figure}

\begin{figure}[tb!]
\centering
\begin{tabular}{@{}c@{}}
\includegraphics[keepaspectratio=true,scale=0.59]{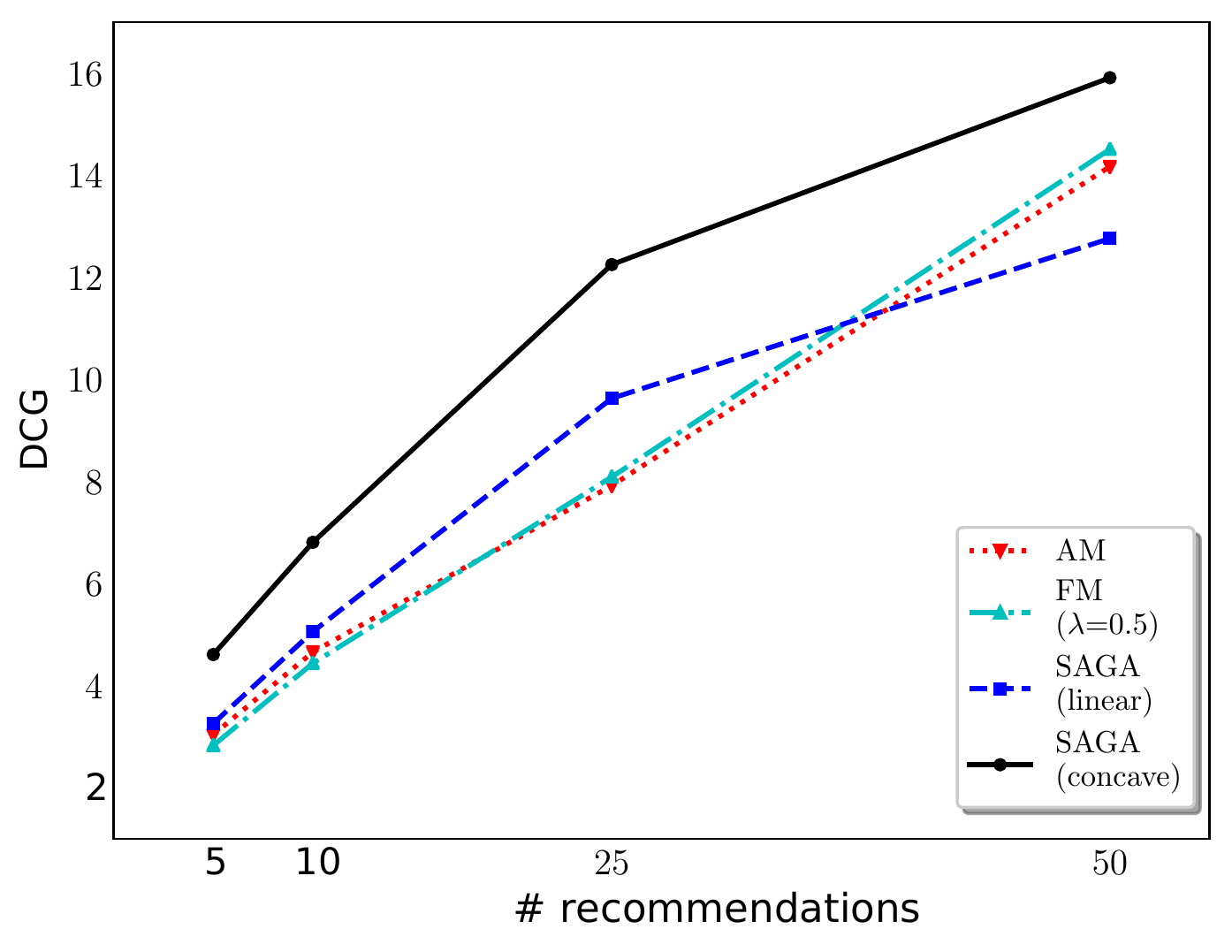}
\end{tabular}
\caption{DCG as function of \# recommendation (similar)}
\label{fig:dcg_sg_rz}
\end{figure}

The results for the similar groups are given in Figures~\ref{fig:dcg_sg_gz} and \ref{fig:dcg_sg_rz}.
Figure~\ref{fig:dcg_sg_gz}  illustrates the performance of different GR algorithms as a function of group size, for a fixed number of recommendations ($\bK$ = 5) for similar groups.
It is very clear from the plot that AM performs better than FM.
Also, as the group size is increased the performance improvement becomes significant.
The linear SAGA algorithm performs marginally better than AM for smaller number of groups, but as more number of users are added to the group, the performance deteriorates.
Unlike in the case of random groups, in case of similar groups the concave formulation of the \textit{user saturation function} gives significant performance boost over AM, FM and SAGA with linear \textit{user saturation function}.
We observed the same pattern by varying the size of the recommendation set.

In Figure~\ref{fig:dcg_sg_rz} the DCG values of different algorithms are plotted as a function of the size of the recommendation set for a fixed group size $(\lvert \calG \rvert) = 4$.
The performance of AM and FM algorithms follows an identitical pattern.
For smaller number of recommendations, AM performs better than FM and as the size of the recommendation set increases, FM outperforms AM.
But in both cases, the significance is negligible.
The linear SAGA algorithm performance degrades below FM and AM as the recommendation size set increases.
The concave SAGA algorithm gives the best DCG values for all recommendation set size and the performance imporovement is significant.
We achieved very similar results by varying the group size.

The performances of coverage metrics for different algorithms are given in Figures~\ref{fig:sr_sg_gz} and \ref{fig:sr_sg_rz}.
Unlike in the case of random groups, all the algorithms perform equally well, and there is no added benefit of using SAGA algorithm if ones motive is only in serendipitous recommendations.
\begin{figure}[tb!]
\centering
\begin{tabular}{@{}c@{}}
\includegraphics[keepaspectratio=true,scale=0.58]{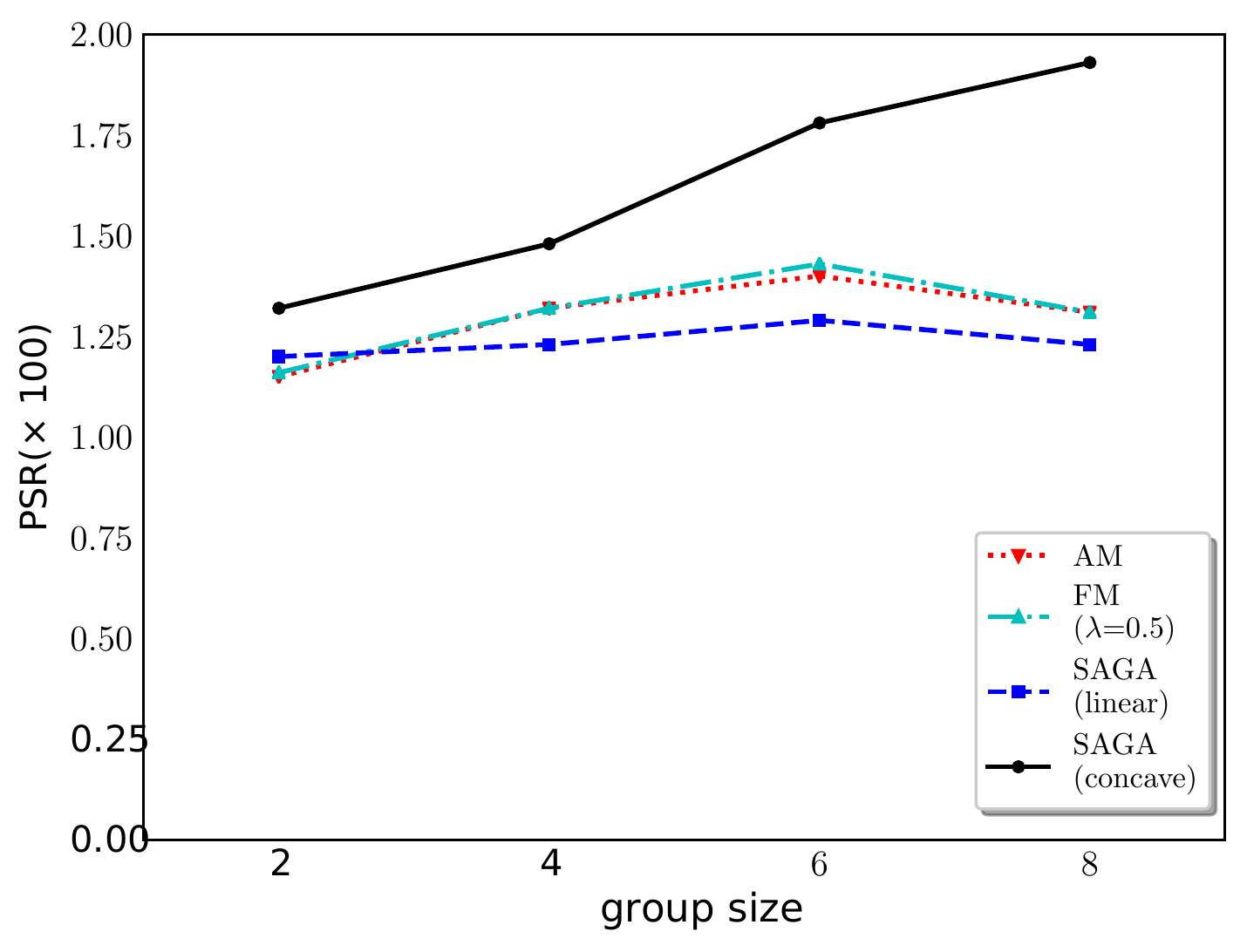}
\end{tabular}
\caption{PSR as function of group size (similar)}
\label{fig:sr_sg_gz}
\end{figure}

\begin{figure}[tb!]
\centering
\begin{tabular}{@{}c@{}}
\includegraphics[keepaspectratio=true,scale=0.59]{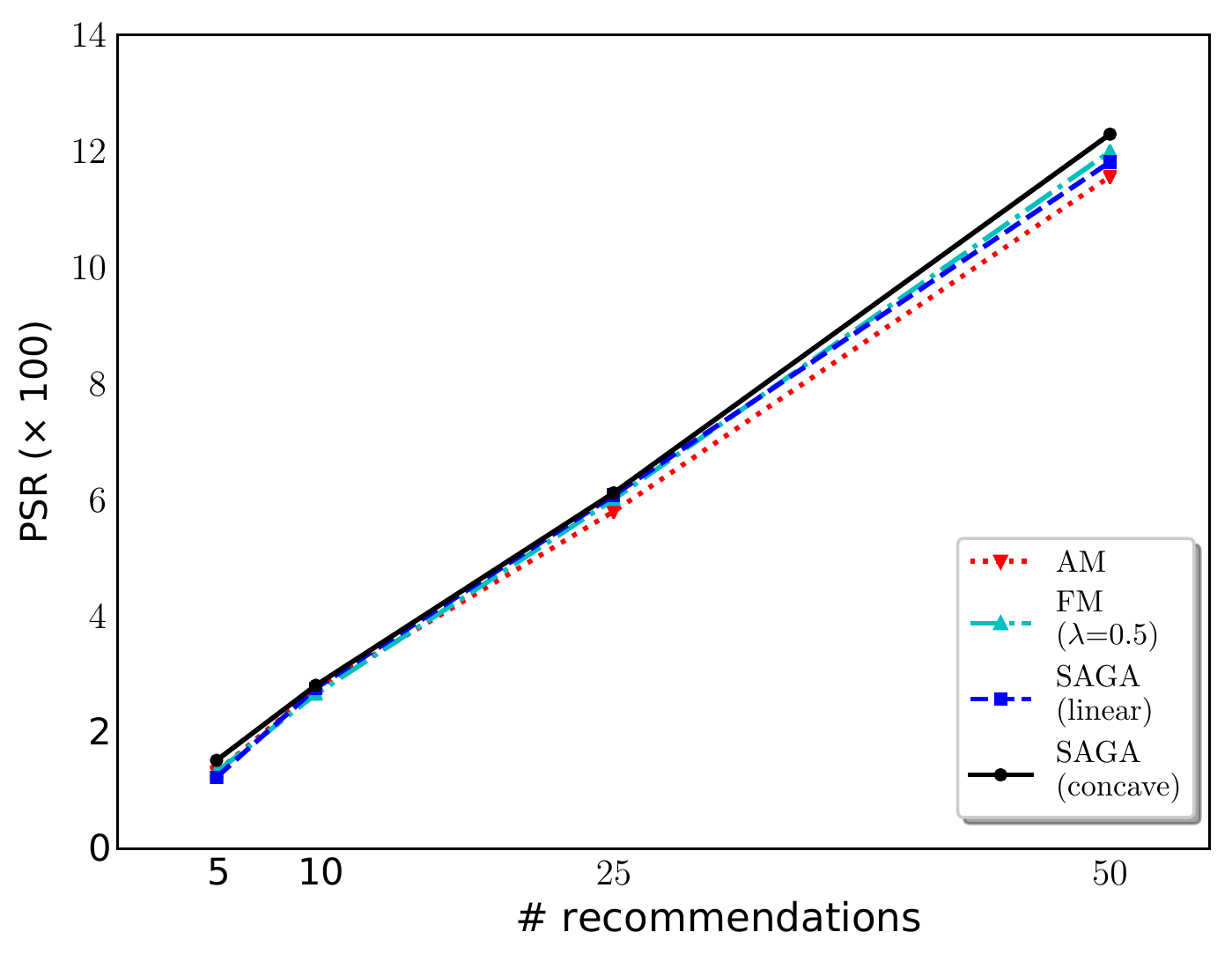}
\end{tabular}
\caption{PSR as function of \# recommendation (similar)}
\label{fig:sr_sg_rz}
\end{figure}

We will now analyze the performance degradation of the concave SAGA algorithm compared to the linear SAGA algorithm for random groups.
In case of random groups, for each group the average user-user similarity value is low (average value is 0.14).
Hence for smaller sized groups the performance gain obtained by employing submodular structure over the \emph{user saturation function} is less than employing the linear structure.
In other words, for smaller groups since the user interests are very different and the relevant movies are often conflicting, adding an item by the diminishing return property does not result in overall performance gain.
We argue that for medium and larger groups, employing a submodular structure does result in performance improvement.
Figure~\ref{fig:dcg_rg_gz8} contains the DCG values for the random group of size 8.
It can be verified that submodularity over \emph{user saturation function} yields better performance.

\begin{figure}[tb!]
\centering
\begin{tabular}{@{}c@{}}
\includegraphics[keepaspectratio=true,scale=0.59]{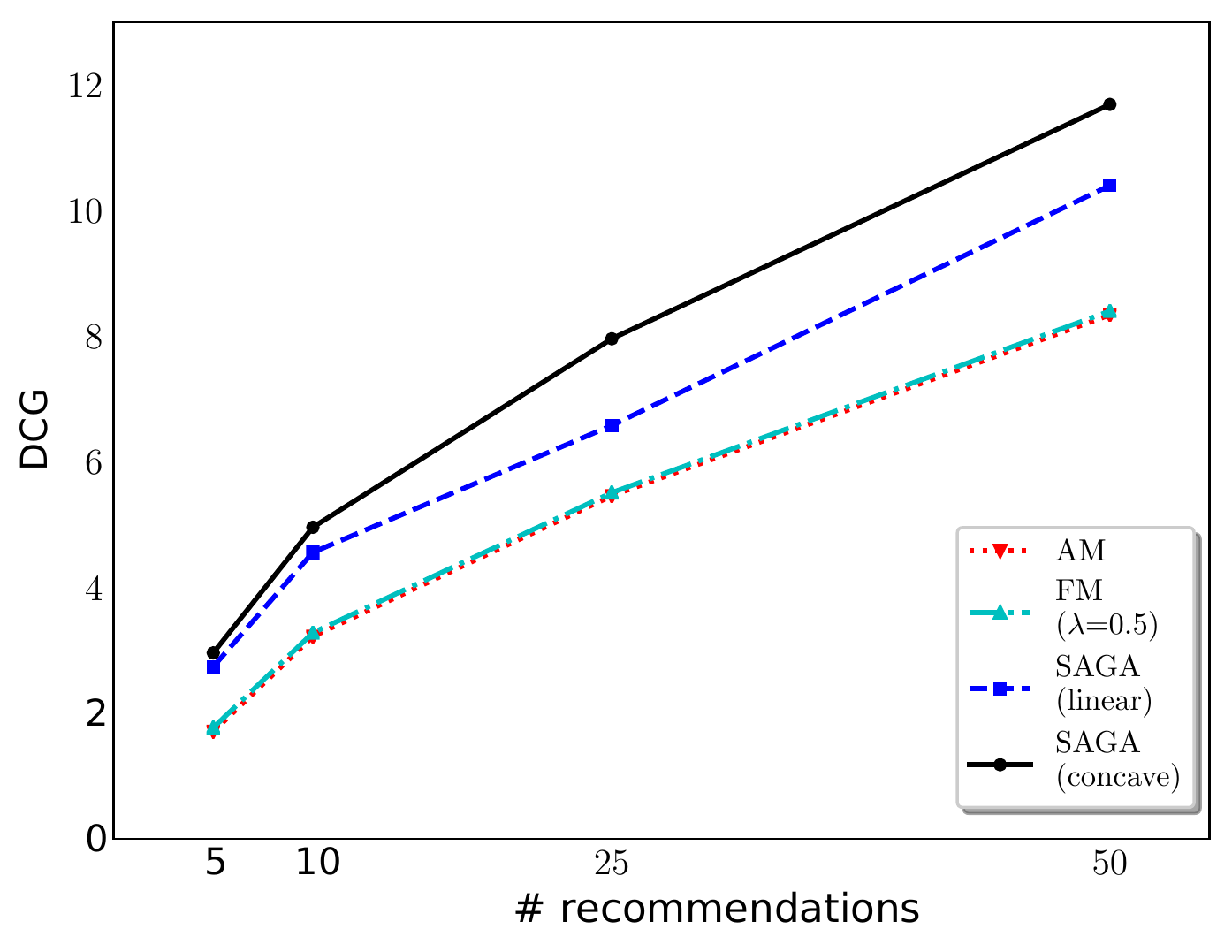}
\end{tabular}
\caption{DCG as function of \# recommendation (random)}
\label{fig:dcg_rg_gz8}
\end{figure}

\section{Conclusion}
\label{sec:conc}
The core idea presented in the paper is that the group recommendation problem can be
modeled as a submodular optimization problem.
In fact, several important modeling options in group recommendation and social choice theory turn out to be specific instances of the group consensus score function we proposed in this work.
The advantage of our approach is that the bundle of items recommended to a  group
not only depends on the aggregate preference of users (in the group) towards the items
but also the affinity (or dissimilarity) between the items.
Experiments on real data sets attest to the efficacy of our approach.

\section{Acknowledgment}
S.P.P thanks Nicolas Usunier and Yves Grandvalet for many fruitful discussions which led us to the core idea presented in this paper.
We thank our anonymous reviewers for their valuable comments and suggestions.

\bibliographystyle{aaai}
\bibliography{paper}
\end{document}